\documentclass[letterpaper, 10 pt, conference]{ieeeconf}

\IEEEoverridecommandlockouts

\usepackage{amsmath}
\usepackage{amssymb}
\usepackage{comment}
\usepackage{amsfonts}

\usepackage{amsthm}
\usepackage{cite}
\usepackage{graphicx}

\usepackage[ruled,vlined]{algorithm2e}

\usepackage[color=green!40]{todonotes}

\theoremstyle{plain}
\newtheorem{assumption}{Assumption}
\newtheorem{remark}{Remark}
\newtheorem{proposition}{Proposition}
\newtheorem{lemma}{Lemma}
\newtheorem{corollary}{Corollary}
\theoremstyle{plain}
\newtheorem{theorem}{Theorem}
\theoremstyle{definition}
\newtheorem{definition}{Definition}
\theoremstyle{definition}

\theoremstyle{definition}
\newtheorem{problem}{Problem}

\title{\LARGE \bf
An Analytical Study of a Two-Sided Mobility Game 
}

\author{Ioannis Vasileios Chremos, \textit{Student Member, IEEE}, and Andreas A. Malikopoulos, \textit{Senior Member, IEEE}%
\thanks{This research was supported by the Sociotechnical Systems Center (SSC) at the University of Delaware.}%
\thanks{The authors are with the Department of Mechanical Engineering, University of Delaware, Newark, DE 19716 USA. {\tt\small{\{ichremos,andreas\}@udel.edu.}}}%
}

\begin{document}

\maketitle
\thispagestyle{empty}
\pagestyle{empty}

\begin{abstract}

In this paper, we consider a mobility system of travelers and providers, and propose a ``mobility game" to study when a traveler is matched to a provider. Each traveler seeks to travel using the services of only one provider, who manages one specific mode of transportation (e.g., car, bus, train, bike). The services of each provider are capacitated and can serve up to a fixed number of travelers at any instant of time. Thus, our problem falls under the category of many-to-one assignment problems, where the goal is to find the conditions that guarantee the stability of assignments. We formulate a linear program of maximizing the social welfare of travelers and providers and show how it is equivalent to the original problem and relate its solutions to stable assignments. We also investigate our results under informational asymmetry and provide a ``mechanism" that elicits the information of travelers and providers. Finally, we investigate and validate the advantages of our method by providing a numerical simulation example.

\end{abstract}

\section{INTRODUCTION}


Commuters in big cities have continuously experienced the frustration of congestion and traffic jams. Travel delays, accidents, and altercations have consistently impacted the economy, society, and the natural environment by playing a decisive role in the increase of idling vehicle engines on city roads \cite{colini_baldeschi2017}. In addition, one of the pressing challenges of our time is the increasing demand for energy, which requires us to make fundamental transformations in how our societies use and access transportation. Thanks to evolutionary developments that are currently afoot, it is highly expected that we will be able to eliminate congestion entirely and significantly increase mobility efficiency in terms of fuel consumption and travel time \cite{zhao2019enhanced}. Self-driving cars offer a most intriguing opportunity that will enable us to travel safely and efficiently anywhere and anytime \cite{barnes2017}. Several studies have shown the benefits of \emph{emerging mobility systems} (e.g., ride-hailing, on-demand mobility services, shared vehicles, autonomous vehicles) to reduce energy and alleviate traffic congestion in a number of different transportation scenarios \cite{sarkar2016,mahbub2020decentralized,zardini2020,Malikopoulos2020}. One question though that still remains unanswered is: \emph{Can we develop an efficient mobility system that can enhance accessibility while controlling the ratio of travel demand over capacity and guarantee the welfare of all travelers?}

Recently it was shown that when daily commuters were offered a convenient and affordable taxi service for their travels, a change of behavior was noticed, namely these commuters ended up using cars more often compared to when they drove their own car \cite{harb2018}. Along with other studies \cite{bissell2018,singleton2018} this shows us that in emerging mobility systems the travelers' tendency to travel will drive travelers to use cars more often and shift them away from public transit. So, in this paper to address this problem, we study the game-theoretic interactions of travelers seeking to travel and getting assigned to providers (e.g., technology and transportation companies) where each operates one mode of transportation. Inspired by the Mobility-as-a-Service concept, we consider a system of multimodal mobility that handles user-centric information and provides travel services (e.g., navigation, location, booking, payment) to a number of travelers. The goal of such a mobility system is to guarantee mobility as a seamless service across all modes of transportation accessible to all. For our purposes, we consider a game to model the strategic and economic interactions in a mobility system with two groups of agents, a traveler and a provider group, where both have preferences and the objective is to assign each traveler to only one provider.


One of the standard approaches to alleviate congestion in a transportation system has been the management of demand size due to the shortage of space availability and scarce economic resources in the form of congestion pricing (alternatively called ``tolling mechanisms" in \cite{pigou2013}). Such an approach focuses primarily on intelligent and scalable traffic routing, in which the objective is to guide and coordinate decision-makers in path-choice decision-making. Interestingly, by adopting a game-theoretic approach, advanced systems have been proposed to assign users concrete routes or minimize travel time and study the Nash equilibria under different tolling mechanisms \cite{wada2011,salazar2019,Chremos2020MechanismDesign,chremos2020SharedMobility,Chremos2020SocialDilemma,chremos2020MobilityMarket}.

Partly related to our work are matching models which describe systems or markets in which there are agents of disjoint groups and have preferences regarding the ``goods" of the opposite agent they associate with. Notable examples are mechanisms for assigning students to schools \cite{abdulkadiroglu2005}, pickup and delivery \cite{Treleaven2011}, electric vehicles \cite{You2017}. It is easy to see that matching markets are quite practical as they offer insights into the more general economic and behavioral real-life situations. These examples are all centralized approaches of determining who gets assigned to whom at what cost and benefit. One of the very first studies was the marriage model which was analyzed by \cite{gale1962} and existence of stable matchings between men and women was established. The authors in \cite{shapley1971} extended it by incorporating monetary transactions between the agents to the marriage model and formulated it to the well-known ``assignment game." They showed that there exists a set of stable assignments, called the \emph{core} (no agent wants to deviate from their match) and it is identical to the solutions of a dual linear program. However, no explicit mechanism was offered on how to achieve a stable assignment in the core. Thanks to the natural usefulness of matching markets, various extensions of the assignment game have been developed focusing either on different behavioral settings or information structures \cite{demange1986,anshelevich2013}.


The main contribution of this paper is the development of a game-theoretical framework to study the economic interactions of travelers and providers in a two-sided many-to-one assignment game. By two-sided we mean that we consider the preferences of both the travelers and the providers. By many-to-one we mean that we impose constraints of how many travelers can be assigned to one provider and how many providers to one traveler. Our analysis can be divided into two parts. First, we use linear programming arguments to showcase the existence of an optimal assignment between travelers and providers that is also stable, i.e., no one will seek to deviate from their match. Second, we consider an asymmetric informational structure, where no traveler/provider is expected to provide their private information willingly. We provide a pricing mechanism (Algorithm \ref{alg:pricing_mechanism}) for this case and show how we can successfully elicit the private information while also ensure efficiency (maximization of social welfare).


The remainder of the paper is structured as follows. In Section \ref{Section:Formulation}, we present the mathematical formulation of the proposed mobility game which forms the basis of our theoretical study for the rest of the paper. In Section \ref{Section:Analysis&Properties}, we derive the theoretical properties of the mobility game and, in Section \ref{Section:SimulationResults}, we validate the game with a numerical simulation. Finally, in Section \ref{Section:Conclusion}, we draw concluding remarks and offer a discussion of some future research directions.

\section{MODELING FRAMEWORK}
\label{Section:Formulation}

\subsection{The ``Mobility Game" Formulation}

We consider a mobility system of two finite, disjoint, and non-empty groups of agents of which one represents the travelers and the other the providers. We denote the set of travelers by $\mathcal{I}$, $|\mathcal{I}| = I \in \mathbb{N}$ and the set of providers by $\mathcal{J}$, $|\mathcal{J}| = J \in \mathbb{N}$. In a typical mobility system, we expect to have more travelers than available providers, so $I \gg J$. Each provider represents a company (e.g., Uber, Lyft, Amtrak, DART, Lime) that manages a fleet of vehicles (cars, trains, busses, bicycles). We focus our study in static settings, in this paper, thus, each provider can serve up to a fixed number of travelers within a fixed time period. For example, in a generic city neighborhood, on any given weekday morning, there are at most a certain number of ride-sharing vehicles available (Uber/Lyft). Formally, for each provider $j \in \mathcal{J}$, we impose a physical \emph{traveler capacity}, denoted by $\varepsilon_j \in \mathbb{N}$. Naturally, each provider can serve a different number of travelers, so we expect $\varepsilon_j$ to vary significantly. For example, a train company can provide travel services per hour to hundreds of travelers compared to a bikeshare company in a city. Next, travelers seek to travel using the services of at most one provider. We do not focus our modeling in routing or path-allocation (such problems have been studied extensively \cite{krichene2015,brown2017}), rather we are interested in an optimal collective assignment of travelers to providers. Both travelers and providers have preferences and can be characterized by their type; thus, this is a two-sided mobility game.

\begin{remark}
    Without loss of generality, we expect the aggregate travel demands of all the travelers to be exactly met by the aggregate capacities of all the providers' mobility services. Thus, we have $\sum_{j \in \mathcal{J}} \varepsilon_j = |\mathcal{I}| = I$.
\end{remark}

\begin{remark}
    Intuitively, via a smartphone app, travelers book in advance for their travel needs and report their preferences and request a travel recommendation (which provider to use). The app collects all requests from specific neighborhoods at a fixed time, and then assigns each traveler to a provider by taking into account both the traveler's as well as the provider's preferences.
\end{remark}

\begin{definition}\label{defn:traveler-provider-assign}
    The \emph{traveler-provider assignment} is a vector $\mathbf{X} = (x_{1 1}, \dots, x_{i j}, \dots, x_{I J}) = (x_{i j})_{{i \in \mathcal{I}}, j \in \mathcal{J}}$, where $x_{i j}$ is a binary variable of the form:
        \begin{equation}
            x_{i j} =
                \begin{cases}
                    1, \; & \text{if $i \in \mathcal{I}$ is assigned to $j \in \mathcal{J}$}, \\
                    0, \; & \text{otherwise}.
                \end{cases}
        \end{equation}
    We call $x_{i j}$ the \emph{mobility outcome} of each traveler $i$ and each provider $j$ and denote by $\mathcal{X}$ the set of such outcomes.
\end{definition}

\begin{definition}
    For any traveler $i \in \mathcal{I}$, $\theta_i = \max_{j \in \mathcal{J}} \, \{ \theta_{i j} \} \in \Theta_i$, where $\theta_{i j} \in [0, 1]$, is traveler $i$'s \emph{personal predisposition} of provider $j \in \mathcal{J}$.
\end{definition}

We denote by $\theta_{- i}$ for the personal predisposition profile of all travelers except traveler $i$. Intuitively, a traveler might have a great affinity towards a taxicab service and a lower affinity towards a bus service. So, we expect different travelers to have different preferences on the mode of transportation to use.

Next, we represent the preferences of each traveler with a utility function consisted of two parts: the traveler's valuation of the mobility outcome and the associated payment required for the realization of the outcome. In other words, any traveler is expected to pay a toll or ticket fee for the services of a provider.

\begin{definition}
    Each traveler $i$'s preferences are summarized by a utility function $u_i : \mathcal{X} \times \Theta_i \to \mathbb{R}$ that determines the monetary value of the overall payoff realized by traveler $i$ from their assignment to provider $j$. Let $t_{i j} \in [\underline{t}, \bar{t}] \subset \mathbb{R}$ denote traveler $i$'s mobility payment. Thus, traveler $i$ receives a total utility in the form
        \begin{equation}\label{eqn:traveler_utility}
            u_i(x_{i j}, \theta_i) = v_i(x_{i j}, \theta_i) - t_{i j},
        \end{equation}
    where $v_i : \mathcal{X} \times \Theta_i \to \mathbb{R}_{\geq 0}$ is a linear valuation function that represents the maximum amount of money that traveler $i$ is willing to pay for the mobility outcome $x_{i j}$.
\end{definition}

\begin{remark}\label{rmk:no_assignment_no_payoff}
    If for any traveler $i \in \mathcal{I}$, we have $x_{i j} = 0$ for all $j \in \mathcal{J}$, then $u_i = 0$. Naturally, this means that for any traveler $i$ with $x_{i j} = 0$ for all $j \in \mathcal{J}$ we have $t_{i j} = 0$.
\end{remark}

On similar lines, we define the providers' utility function.

\begin{definition}
    A provider $j$'s utility is given by
        \begin{equation}\label{eqn:provider_utility}
    	    u_j(x_{i j}, \delta_j) = t_{i j} - c_j(x_{i j}, \delta_j),
	    \end{equation}
where $\delta_j \in (0, 1]$ represents the type of provider $j$, and $c_j$ is a linear cost function related to the operation of the mobility services provided by $j \in \mathcal{J}$. We denote by $\delta_{- j}$ for the type profile of all providers except provider $j$.
\end{definition}

\begin{remark}
    Intuitively, $\delta_j$ can be interpreted as the ``operational value" of provider $j$ for the mobility services it provides and operates. In other words, the monetary value of the entire process of its service to serve a traveler on a given location and time.
\end{remark}

In both \eqref{eqn:traveler_utility} and \eqref{eqn:provider_utility}, the ``payment component" $t_{i j}$ is not expected to dominate either the traveler's or the provider's utility function. This is because $t_{i j}$ have an alternate sign in \eqref{eqn:traveler_utility} and \eqref{eqn:provider_utility}, so a high value (or low) can lead to negative utility for the travelers (or the providers) leading to a unfavorable match between traveler $i$ and provider $j$. We will see later in the paper how we can ensure unfavorable matchings do not happen.

\begin{definition}
    Under the assignment $x_{i j}$ of traveler $i$ and provider $j$, their \emph{mobility $(i, j)$-matching payoff} is given by
        \begin{equation}
            a_{i j}(x_{i j}) = u_i(x_{i j}, \theta_i) + u_j(x_{i j}, \delta_j),
        \end{equation}
    where $a_{i j}$ measures the combined payoff or benefit measured in monetary units of traveler $i$ being assigned to provider $j$.
\end{definition}

\begin{remark}
    By Remark \ref{rmk:no_assignment_no_payoff}, if $x_{i j} = 0$, then $a_{i j} = 0$.
\end{remark}

\begin{definition}\label{defn:utility_matrix}
    The \emph{utility assignment matrix} $\mathbf{A}$ is constructed with $|\mathcal{I}|$ rows and $|\mathcal{J}|$ columns and each entry represents the $(i, j)$-matching utility $a_{i j}$ between traveler $i$ and provider $j$ for all $i \in \mathcal{I}$ and all $j \in \mathcal{J}$.
\end{definition}

Based on Definition \ref{defn:utility_matrix}, we can construct matrix $\mathbf{A}$ as follows:
    \begin{equation}
        \mathbf{A} =
            \begin{bmatrix}
                a_{11} & a_{12} & a_{13} & \dots  & a_{1J} \\
                a_{21} & a_{22} & a_{23} & \dots  & a_{2J} \\
                \vdots & \vdots & \vdots & \ddots & \vdots \\
                a_{I1} & a_{I2} & a_{I3} & \dots  & a_{IJ}
            \end{bmatrix}.
    \end{equation}

The mobility game of travelers and providers is a collection of four objects, namely set of agents, vector of assignments, matrix of utilities, and a vector of capacities. Formally, we state the next definition.

\begin{definition}
    The mobility game can be fully characterized by the tuple $\mathcal{M} = \langle \mathcal{I} \cup \mathcal{J}, \mathbf{X} = (x_{i j})_{i \in \mathcal{I}, j \in \mathcal{J}}, \mathbf{A}, (\varepsilon_j)_{j \in \mathcal{J}} \rangle$.
\end{definition}

\begin{definition}\label{defn:feasibility_stability}
    A \emph{feasible} assignment is a vector $\mathbf{X} = (x_{i j})_{{i \in \mathcal{I}}, j \in \mathcal{J}}$, $x_{i j} \in \{0, 1\}$ that satisfies constraints
        \begin{align}
            \sum_{j \in \mathcal{J}} x_{i j} & \leq 1, \quad \forall i \in \mathcal{I}, \label{Constraint:Problem1-First} \\
            \sum_{i \in \mathcal{I}} x_{i j} & \leq \varepsilon_j, \quad \forall j \in \mathcal{J}, \label{Constraint:Problem1-Second}
        \end{align}
    where \eqref{Constraint:Problem1-First} ensures that each traveler $i \in \mathcal{I}$ is assigned to only one mobility service $j \in \mathcal{J}$, and \eqref{Constraint:Problem1-Second} ensures that the traveler capacity of each provider $j$ is not exceeded while its services are shared by multiple travelers. An \emph{optimal} assignment is a feasible assignment $(x_{i j})_{{i \in \mathcal{I}}, j \in \mathcal{J}}$ such that
        \begin{equation}\label{eqn:stability_condition}
            \sum_{i \in \mathcal{I}} \sum_{j \in \mathcal{J}} a_{i j}(x_{i j}) \geq \sum_{i \in \mathcal{I}} \sum_{j \in \mathcal{J}} a_{i j}(x_{i j} '),
        \end{equation}
    for all feasible assignments $x_{i j} '$.
\end{definition}

\begin{definition}\label{defn:stability}
    A feasible assignment $\mathbf{X} = (x_{i j})_{{i \in \mathcal{I}}, j \in \mathcal{J}}$, $x_{i j} \in \{0, 1\}$ is \emph{stable} if there exist non-negative vectors $\phi = (\phi_i)_{i \in \mathcal{I}}$ and $\psi = (\psi_j)_{j \in \mathcal{J}}$ such that
	    \begin{equation}
		    \sum_{i \in \mathcal{I}} \phi_i + \sum_{j \in \mathcal{J}} \varepsilon_j \cdot \psi_j = \sum_{i \in \mathcal{I}} \sum_{j \in \mathcal{J}} a_{i j}(x_{i j})
	    \end{equation}
    with $\phi_i + \psi_j \geq a_{i j}$ for all $i \in \mathcal{I}$ and all $j \in \mathcal{J}$.
\end{definition}

We will see later in Section \ref{Section:Analysis&Properties} the mathematical and physical interpretation of $\phi$ and $\psi$.

\begin{definition}\label{defn:ideal_equilibrium}
	Let $(t_{i j} ^ *)_{i \in \mathcal{I}, j \in \mathcal{J}}$ denote the mobility payments associated with the stable assignment, denoted by $(x_{i j} ^ *)_{{i \in \mathcal{I}}, j \in \mathcal{J}}$. Then the equilibrium $(x_{i j} ^ *, t_{i j} ^ *)_{i \in \mathcal{I}, j \in \mathcal{J}}$ is called an \emph{ideal-mobility} equilibrium.
\end{definition}

From Definition \ref{defn:stability}, it is easy to see that Definition \ref{defn:ideal_equilibrium} implies that an ideal-mobility equilibrium in mobility game $\mathcal{M}$ ensures that (i) providers are assigned to travelers up to their capacity (thus, maximizing revenue), and (ii) travelers receive the best-possible utility being assigned to a provider (thus, maximizing welfare).


\begin{assumption}\label{ass:complete_information}
	Every aspect of the mobility game $\mathcal{M}$ is considered known information to every traveler and provider.
\end{assumption}

Assumption \ref{ass:complete_information} seems strong but it will prove instrumental in Section \ref{Section:Analysis&Properties}. Our analysis will focus on how to show existence, optimality, and stability of traveler-provider assignments and then in Subsection \ref{subsection:mechanism}, we will relax Assumption \ref{ass:complete_information} and show how we can elicit the private information of both travelers and providers.

\subsection{The Optimization Problem}
\label{Subsection:ProblemFormulation}

In the mobility game $\mathcal{M}$, we are interested to know what are its stable assignments (alternatively called stable equilibria), whether they exist and under what conditions.

\begin{problem}
\label{pro:primal}
    The maximization problem of $\mathcal{M}$ is
        \begin{gather}
            \max_{x_{i j}} \sum_{i \in \mathcal{I}} \sum_{j \in \mathcal{J}} a_{i j}(x_{i j}), \label{Equation:Problem1-ObjectiveFunction} \\
            \text{subject to:} \notag \;
            \eqref{Constraint:Problem1-First}, \eqref{Constraint:Problem1-Second},
        \end{gather}
    where $x_{i j} \in \{0, 1\}$ for all $i \in \mathcal{I}$ and all $j \in \mathcal{J}$.
\end{problem}

We can relax the binary variable constraint to a non-negativity constraint variable in Problem \ref{pro:primal}. We will show in the next section that this does not affect the optimal solutions of Problem \ref{pro:primal} as we can ensure all optimal solutions of the equivalent linear program are binary valued.

\begin{problem}[Linear Program]
\label{prob:linear}
    The linear program formulation of mobility game $\mathcal{M}$ is
        \begin{gather}
            \max_{x_{i j}} \sum_{i \in \mathcal{I}} \sum_{j \in \mathcal{J}} a_{i j}(x_{i j}) \label{Equation:Problem2-ObjectiveFunction} \\
            \text{subject to:} \notag \;
            \eqref{Constraint:Problem1-First}, \eqref{Constraint:Problem1-Second}, \text{ and} \\
            x_{i j} \geq 0, \quad \forall i \in \mathcal{I}, \quad \forall j \in \mathcal{J}, \label{Constraint:Problem2-Third}
        \end{gather}
    where \eqref{Constraint:Problem2-Third} transforms the (binary) assignment problem to a (continuous) linear program of which $x_{i j}$ can be interpreted as the probability that traveler $i$ is matched to provider $j$.
\end{problem}

\section{ANALYSIS AND PROPERTIES OF THE MOBILITY GAME}
\label{Section:Analysis&Properties}

\subsection{Existence, Optimality, and Stability of Assignments}

In this subsection, we show that for Problems \ref{pro:primal} and \ref{prob:linear} at least one optimal solution exists (thus, ensuring stability).

\begin{theorem}
	The stable assignments of mobility game $\mathcal{M}$ are the same with the optimal solutions of Problem \ref{pro:primal}. Furthermore, the set of optimal solutions of Problem \ref{pro:primal} is non-empty.
\end{theorem}

\begin{proof}
    By relaxing the binary constraint of Problem \ref{pro:primal}, we get a linear program (Problem \ref{prob:linear}). Its set of all real-valued solutions is a polytope whose vertices have all integer-valued coordinates. Since the solutions are also guaranteed to be non-negative, the set of solutions is non-empty \cite{schrijver1996}. Thus, Problem \ref{prob:linear} has at least one solution with integer components (in our case 0-1 components). Hence, the set of all optimal solution of Problem \ref{pro:primal} is non-empty \cite{schrijver1996}. By Definition \ref{defn:feasibility_stability}, assignments are stable as long as no agent in $\mathcal{I} \cup \mathcal{J}$ has an incentive (e.g., higher utility) to break their matching pair. So, finding a stable assignment is equivalent to finding the best in terms of aggregate utility among all possible feasible assignments. Mathematically, \eqref{eqn:stability_condition} naturally leads to a maximization problem. Therefore, the existence of a stable assignment of mobility game $\mathcal{M}$ is guaranteed.
\end{proof}

Next, we derive the dual of Problem \ref{prob:linear}.

\begin{problem}
\label{prob:dual}
    The dual of Problem \ref{prob:linear} is given below:
        \begin{gather}
            \min_{\phi, \psi} \sum_{i \in \mathcal{I}} \phi_i + \sum_{j \in \mathcal{J}} \varepsilon_j \cdot \psi_j, \label{Equation:Problem3-ObjectiveFunction} \\
            \text{subject to:} \notag \\
            \phi_i + \psi_j \geq a_{i j}, \quad \forall i \in \mathcal{I}, \quad \forall j \in \mathcal{J}, \label{Constraint:Problem3-First} \\
            \phi_i \geq 0, \quad \forall i \in \mathcal{I}, \label{Constraint:Problem3-Second} \\
            \psi_i \geq 0, \quad \forall j \in \mathcal{J}, \label{Constraint:Problem3-Third}
        \end{gather}
    where $\phi$ is a $|\mathcal{I}|$-dimensional vector and $\psi$ is a $|\mathcal{J}|$-dimensional vector.
\end{problem}

Our objective is to establish a method for the mobility game $\mathcal{M}$'s stable assignments by solving Problem \ref{prob:linear}. In turn, we want to solve Problem \ref{prob:dual} to find the stable assignments. This is possible only if we can guarantee strong duality (satisfying the conditions of complementary slackness). Formally, a feasible assignment $x_{i j}$ and a feasible solution $(\phi, \psi)$ are optimal if and only if
	\begin{equation}\label{eqn:strong_duality}
		\sum_{i \in \mathcal{I}} \sum_{j \in \mathcal{J}} a_{i j}(x_{i j}) = \sum_{i \in \mathcal{I}} \phi_i + \sum_{j \in \mathcal{J}} \varepsilon_j \cdot \psi_j.
	\end{equation}
The conditions that guarantee optimality are given by the theorem of complementary slackness, i.e.,
	\begin{align}
		\phi_i + \psi_j - a_{i j} & = 0, \quad \forall i \in \mathcal{I}, \quad \forall j \in \mathcal{J}, \label{eqn:complementary_slackness_1} \\
		\sum_{j \in \mathcal{J}} (x_{i j} -1) \cdot \phi_i & = 0, \quad \forall i \in \mathcal{I}, \label{eqn:complementary_slackness_2} \\
		\sum_{i \in \mathcal{I}} (x_{i j} - \varepsilon_j) \cdot \psi_j & = 0, \quad \forall j \in \mathcal{J}. \label{eqn:complementary_slackness_3}
	\end{align}

\begin{lemma}\label{lem:convexity_dual}
    The set of solutions of Problem \ref{prob:dual} is non-empty and convex.
\end{lemma}

\begin{proof}
	We have already established that Problem \ref{prob:linear} has at least one solution. Thus, it follows easily that Problem \ref{prob:dual} has at least one solution too. Any solution of Problem \ref{prob:dual} has a specific structure due to the geometry of the constraint set \eqref{Constraint:Problem3-First} - \eqref{Constraint:Problem3-Third}. Since at least one solution will be in 0-1 components, the constraints will force this solution to be in at a corner of a polyhedra. Thus, the set of solutions of Problem \ref{prob:dual} is non-empty and has to be convex.
\end{proof}

\begin{remark}\label{rmk:interpretation_dual_variables}
    Intuitively, a dual solution $(\phi, \psi)$ can be seen as a method to share the ``gains of mobility" among travelers and providers at an ideal-mobility equilibrium (see Definition \ref{defn:ideal_equilibrium}). For example, component of vector $\phi$ describes the realized gain of traveler $i$ when assigned to provider $j$ (thus enjoying the mobility services of provider $j$). A component of vector $\psi$ describes the per unit gain of provider $j$.
\end{remark}

\begin{corollary}
    The set of solutions of Problem \ref{prob:dual} is a compact subset of $\mathbb{R} ^ {|\mathcal{I}|} \times \mathbb{R} ^ {|\mathcal{J}|}$.
\end{corollary}

\begin{proof}
	By Lemma \ref{lem:convexity_dual} and Remark \ref{rmk:interpretation_dual_variables}, it is straightforward to show that the set of solutions of Problem \ref{prob:dual} is compact.
\end{proof}

\begin{corollary}
	There always exists at least one profile of mobility payments $(t_{i j})_{i \in \mathcal{I}, j \in \mathcal{J}}$ under assignment $(x_{i j})_{i \in \mathcal{I}, j \in \mathcal{J}}$.
\end{corollary}

\begin{proof}
    By definition of the mobility game $\mathcal{M}$ for any (feasible) assignment $(x_{i j})_{i \in \mathcal{I}, j \in \mathcal{J}}$, there must be an associated profile of mobility payments $(t_{i j})_{i \in \mathcal{I}, j \in \mathcal{J}}$.
\end{proof}

Next, we show that the existence of an optimal profile of mobility payments $(t_{i j})_{i \in \mathcal{I}, j \in \mathcal{J}}$ can be guaranteed by the formulation of the dual program of Problem \ref{prob:linear} and the computation of its solutions.

\begin{theorem}\label{thm:existence_payment}
	There exists an optimal profile of mobility payments $(t_{i j} ^ *)_{i \in \mathcal{I}, j \in \mathcal{J}}$ under stable assignment $(x_{i j} ^ *)_{i \in \mathcal{I}, j \in \mathcal{J}}$. Furthermore, we must have $\phi_i = u_i$ and $\psi_j = u_j$ for all $i \in \mathcal{I}$ and all $j \in \mathcal{J}$.
\end{theorem}

\begin{proof}
	Suppose $\mathbf{x} ^ * = (x_{i j} ^ *)_{i \in \mathcal{I}, j \in \mathcal{J}}$ is a stable assignment for mobility game $\mathcal{M}$. Under $\mathbf{x} ^ *$, we can calculate $\sum_{i \in \mathcal{I}} \sum_{j \in \mathcal{J}} a_{i j}(x_{i j})$, which by the theorem of strong duality and the definition of stability, \eqref{eqn:strong_duality} holds true. This is because \eqref{eqn:complementary_slackness_1} - \eqref{eqn:complementary_slackness_3} are equivalent to the conditions that ensure stability. Thus, there exist vectors $(\phi ^ *, \psi ^ *)$ from Problem \ref{prob:dual} that feasible and optimal. By Definition \ref{defn:ideal_equilibrium}, it follows that the mobility payments associated with the optimal assignments of travelers and providers are essentially the same with the optimal solutions of Problems \ref{prob:linear} and \ref{prob:dual}. Therefore, by the established existence of solutions to Problems \ref{prob:linear} and \ref{prob:dual} as long as there exists a stable assignment $(x_{i j} ^ *)_{i \in \mathcal{I}, j \in \mathcal{J}}$ an optimal profile of mobility payment $(t_{i j} ^ *)_{i \in \mathcal{I}, j \in \mathcal{J}}$ must exist. By Remark \ref{rmk:interpretation_dual_variables} and Definition \ref{defn:stability} at an optimal assignment we have $\phi_i = u_i$ and $\psi_j = u_j$ for all $i \in \mathcal{I}$ and all $j \in \mathcal{J}$.
\end{proof}

\subsection{Asymmetric Information in the Mobility Game}
\label{subsection:mechanism}

So far, we have implicitly assumed that both the travelers and providers have complete information of the entire information structure of the mobility game $\mathcal{M}$. In other words, each traveler knows every other travelers' and providers' information, i.e., travelers know each others' utilities and valuations, providers know each other providers' types and cost functions. In a realistic setting, this implicit assumption is unreasonably restrictive. Thus, for the rest of the paper, we focus on a ``mechanism" that \emph{induces the mobility game} $\mathcal{M}$ by eliciting the private information of all the travelers and providers. First, we relax Assumption \ref{ass:complete_information} and consider that the types of travelers, i.e., $\theta = (\theta_i)$ and of providers, i.e., $\delta = (\delta_j)$ are private information, i.e., known only to themselves. Next, we denote by $\mathbf{X}_{- i}$ the assignment of travelers in $\mathcal{I} \setminus \{i\}$ to providers in $\mathcal{J}$. Similarly, we denote by $\mathbf{X}_{- j}$ the assignment of travelers in $\mathcal{I}$ to providers in $\mathcal{J} \setminus \{j\}$. Furthermore, we assume that travelers are charged by the mechanism, say $t_i \in \mathbb{R}$, and providers are compensated by the mechanism, say $t_j \in \mathbb{R}$. The proposed mechanism ensures to collect all funds from the travelers and compensate accordingly the providers.

\begin{algorithm}[ht]\label{alg:pricing_mechanism}
	\SetAlgoLined
	\KwData{$\mathcal{I}, \mathcal{J}, (\theta_i)_{i \in \mathcal{I}}, (\delta_j)_{j \in \mathcal{J}}$}
	\KwResult{$\mathbf{x} ^ *$ and $(t_{i j})_{i \in \mathcal{I}, j \in \mathcal{J}}$}
	Define the valuation functions of every traveler and provider and use them to construct matrix $\mathbf{A}$. Solve for the optimal solution $\mathbf{x} ^ *$ of Problem \ref{pro:primal}\;
	\For{$i \in \mathcal{I}$}{
		Solve for the optimal solution $\mathbf{X}_{- i} ^ *$ of Problem \ref{pro:primal}\;
		Set the mobility payment for each traveler $i$:
			\begin{multline*}
				t_i = \sum_{\ell \in \mathcal{I} \setminus \{i\}} \sum_{j \in \mathcal{J}} u_{\ell}(x_{i j}, \theta_{- i}) \\
				- \sum_{\ell \in \mathcal{I} \setminus \{i\}} \sum_{j \in \mathcal{J}} u_{\ell}(x_{i j}, \theta_{\ell})
			\end{multline*}
	}
	\For{$j \in \mathcal{J}$}{
		Solve for the optimal solution $\mathbf{x}_{- j} ^ *$ of Problem \ref{pro:primal}\;
		Set the mobility payment for each provider $j$:
		    \begin{multline*}
		        t_j = \sum_{i \in \mathcal{I}} \sum_{\kappa \in \mathcal{J} \setminus \{j\}} u_{\kappa}(x_{i j}, \delta_{- j}) \\
		        - \sum_{i \in \mathcal{I}} \sum_{\kappa \in \mathcal{J} \setminus \{j\}} u_{\kappa}(x_{i j}, \delta_{\kappa})
		    \end{multline*}
	}
	\caption{Pricing Mechanism}
\end{algorithm}

\begin{theorem}[Voluntary Participation]
	No traveler $i \in \mathcal{I}$ and no provider $j \in \mathcal{J}$ can gain for better individual utility by matching externally compared to the utility gained by participating in the induced mobility game $\mathcal{M}$.
\end{theorem}

\begin{proof}
	It is sufficient to show that no agent in $\mathcal{I} \cup \mathcal{J}$ can gain negative utility by participating in the induced game $\mathcal{M}$, i.e., we must have $u_i, u_j \geq 0$ for all $i, j \in \mathcal{I} \cup \mathcal{J}$. First, note that the maximization of $\sum_{i \in \mathcal{I}} \sum_{j \in \mathcal{J}} a_{i j}(x_{i j})$ is the highest possible value we can achieve. Removing even one agent, does not increase this value under any scenario. Thus, by definition, both payments $t_i$ and $t_j$ are non-negative. At equilibrium, the utilities of any traveler and provider are equivalent to the solutions $(\phi, \psi)$ of Problem \ref{prob:dual} (as we showed in Theorem \ref{thm:existence_payment}). Since $(\phi, \psi)$ ensures non-negativity it follows that $u_i, u_j \geq 0$ for all $i, j \in \mathcal{I} \cup \mathcal{J}$.
\end{proof}

\begin{theorem}[Truthfulness]
	Misreporting does not benefit any traveler or provider.
\end{theorem}

\begin{proof}
	Let us assume that all agents in $\mathcal{I} \cup \mathcal{J}$ voluntary participate in the mechanism. We consider two cases when traveler $i$ misreports their true type, i.e., $\hat{\theta}_i \geq \theta_i$ and $\hat{\theta}_i \leq \theta_i$, where $\hat{\theta}_i$ is traveler $i$'s report. If traveler $i$ reports $\hat{\theta}_i$ that is lower than their true type, then traveler $i$ cannot improve their utility as $t_i$ is necessarily non-negative and a lower misreporting $\hat{\theta}_i \leq \theta_i$ can only lead to lower utilities. Suppose now that traveler $i$ reports $\hat{\theta}_i$ that is higher than their true type. The exact value of $\hat{\theta}_i$ can be chosen by the following maximization problem
	    \begin{equation}
	        \max_{\hat{\theta}_i} u_i(x_{i j}, \hat{\theta}_i) = \max_{\hat{\theta}_i} v_i(x_{i j}, \hat{\theta}_i) - t_i,
	    \end{equation}
	where $t_i$ is given by Algorithm \ref{alg:pricing_mechanism}. Note though that the maximization of $v_i(x_{i j}, \hat{\theta}_i) - \sum_{\ell \in \mathcal{I} \setminus \{i\}} \sum_{j \in \mathcal{J}} u_{\ell}(x_{i j}, \hat{\theta}_{- i}) - \sum_{\ell \in \mathcal{I} \setminus \{i\}} \sum_{j \in \mathcal{J}} u_{\ell}(x_{i j}, \hat{\theta}_{\ell})$ is equivalent to the maximization of $\sum_{i \in \mathcal{I}} \sum_{j \in \mathcal{J}} v_i(x_{i j}, \hat{\theta}_i)$ with respect to the assignment $x_{i j}$. After all that is the goal of traveler $i$, namely by misreporting their type to lead the mechanism to a better assignment. Let $\hat{x}_{i j} ^ *$ and $x_{i j} ^ *$ denote the optimal assignment with $\hat{\theta}_i$ and $\theta_i$, respectively. Hence, we have
	    \begin{equation}\label{eqn:last_argument_IC}
	        \arg \max_{\hat{x}_{i j} ^ *} \sum_{i \in \mathcal{I}} \sum_{j \in \mathcal{J}} v_i(\hat{x}_{i j} ^ *, \hat{\theta}_i) \leq \arg \max_{x_{i j} ^ *} \sum_{i \in \mathcal{I}} \sum_{j \in \mathcal{J}} v_i(x_{i j} ^ *, \theta_i)
	    \end{equation}
	for each traveler $i$. It follows immediately from \eqref{eqn:last_argument_IC} that traveler $i$ can maximize their utility by minimizing the mobility payments as defined in Algorithm \ref{alg:pricing_mechanism}. This is only possible when $\hat{\theta}_i = \theta_i$. Thus, traveler $i$ cannot improve their utility by misreporting. Therefore, we conclude that, with the proposed pricing mechanism (Algorithm \ref{alg:pricing_mechanism}), under no circumstance can traveler $i$ improve their utility by misreporting about their type $\theta_i$. In other words, any traveler $i$ has a strategy to always truthfully report their type to the mechanism. We can follow the same arguments to show this for the providers. Therefore, the proof is completed.
\end{proof}

\begin{proposition}
	If traveler $i$ is matched to provider $j$ while having misreported their type to the mechanism, then traveler $i$ does not gain a better utility compared to the utility gained under the true type.
\end{proposition}

\begin{proof}
	We show this only for the travelers as the arguments are similar for the providers. By construction of the mobility payments in Algorithm \ref{alg:pricing_mechanism} non-negativity of the payments for each traveler is guaranteed, i.e., $t_i \geq 0$. By definition, the valuation of each traveler is non-negative under any assignment. Thus, the utility defined in \eqref{eqn:traveler_utility} is also non-negative. At equilibrium, we have
	    \begin{equation}
	        \max_{x_{i j}} \sum_{i \in \mathcal{I}} v_i(x_{i j}, \theta_i) \leq \sum_{i \in \mathcal{I}} v_i(x_{i j} ^ *).
	    \end{equation}
    Thus, it follows that
        \begin{align}
            u_i(x_{i j}, \theta_i) & = v_i(x_{i j}, \theta_i) - t_i \notag \\
            & = \sum_{i \in \mathcal{I}} v_i(x_{i j} ^ *) - \sum_{\ell \in \mathcal{I} \setminus \{i\}} v_{\ell}(x_{i j}, \theta_{\ell}) \geq 0.
        \end{align}
    Therefore, no traveler can hope for better utility by misreporting.
\end{proof}

\begin{proposition}[Social Efficiency]
	The proposed mechanism satisfies social efficiency as it ensures the maximization of the aggregate social welfare of both travelers and providers.
\end{proposition}

\begin{proof}
	By construction of the mobility payments in Algorithm \ref{alg:pricing_mechanism}, it follows immediately that the optimal solution maximizes the social welfare.
\end{proof}

\section{SIMULATION RESULTS}
\label{Section:SimulationResults}

In this section, we present a numerical example, its solution and discuss its physical interpretation. Consider a mobility system of four providers $\mathcal{J} = \{\text{bike}, \text{car}, \text{bus}, \text{train}\}$ and twenty travelers $\mathcal{I} = \{1, 2, \dots, 20\}$. Each provider has traveler capacities, namely we have $\varepsilon_{bike} = 1$, $\varepsilon_{car} = 4$, $\varepsilon_{bus} = 5$, and $\varepsilon_{train} = 10$. Moreover, we partition the set of travelers $\mathcal{I}$ into four types, i.e., students, commuters, tourists, consumers with sizes $|\mathcal{I}_{\text{students}}| = 3$, $|\mathcal{I}_{\text{commuters}}| = 5$, $|\mathcal{I}_{\text{tourists}}| = 4$, $|\mathcal{I}_{\text{consumers}}| = 8$. Each type of travelers can represent the personal predisposition $\theta = (\theta_i)_{i \in \mathcal{I}}$. Next, with a slight abuse of notation, we generate a random utility assignment matrix
	\begin{equation}
		\mathbf{A} = \\
			\begin{bmatrix}
				1 & 2 & 0.5 \\
				2.5 & 2 & 1.5 \\
				2.5 & 4 & 1.5 \\
				2.5 & 5 & 6.5
			\end{bmatrix},
	\end{equation}
where each row represents a type of travelers and each column represents a provider. The entry $a_{i j}$ of $\mathbf{A}$ represents the overall utility of assignment $x_{i j}$.

We solve Problem \ref{prob:linear} and compute with an optimal solution that maximizes the aggregate utilities of each traveler and each provider according to the travelers' preferences and maximizing the capacities of each provider. The computational complexity of the proposed method is relatively low as long as the number of travelers and providers remain small. This is reasonable to expect as at any given moment there can only less than five different travel options (so, the number of providers is always small). By ensuring we partition the travelers' requests according to origin, destination, and type, we can make certain that the number of travelers does not make the optimization problems untrackable.

\begin{figure}[ht]
	\centering
	\includegraphics[width = 1 \columnwidth]{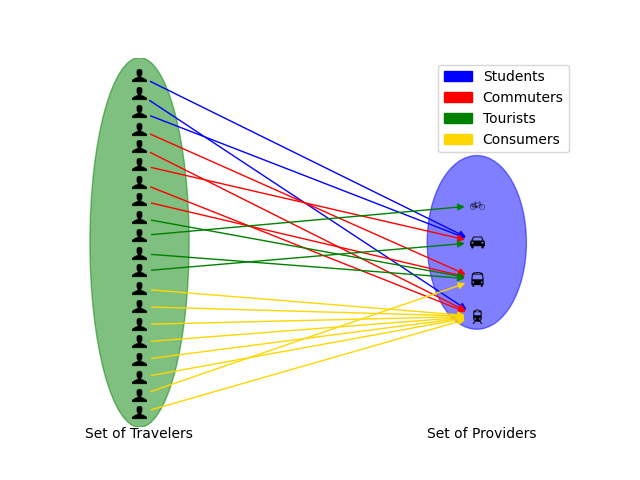}
	\caption{Optimal assignment of travelers to providers with 20 travelers and 4 providers.}
	\label{fig:OptimalAssignment}
\end{figure}

We can see from Fig. \ref{fig:OptimalAssignment} that an efficient allocation of the providers' resources and services to different types of travelers can be attained using a game-theoretic framework. The assignment shown in Fig. \ref{fig:OptimalAssignment} is stable, maximizes the social welfare, and maximizes the capacity of the providers (thus, maximizing their utilities too).

\section{CONCLUDING REMARKS}
\label{Section:Conclusion}

In this paper, we provided a theoretical study of a two-sided game for a mobility system of travelers and providers focusing on how to discretely assign travelers to providers when both have preferences. We formulated a binary program and its equivalent linear program and showed that at least one optimal solution exists and derived the conditions for such solution to be stable. We then allowed informational asymmetry in the proposed mobility game and provided a pricing mechanism to ensure we can elicit the private information of all travelers and providers. We showed that our mechanism guarantees economic efficiency in terms of maximizing the social welfare, and ensures voluntary participation, thus making sure that all agents have a unique dominant strategy.

Ongoing work includes relaxing our assumption of linearity in the utility functions and also investigating our model under the behavioral model of \emph{prospect theory} \cite{tversky1992}. An interesting research direction should involve an extension of our model with a sophisticated construction of the travelers' utilities and preferences using data gathered from a behavioral survey. This would helps us observe any correlations between the travelers' tendencies or attitudes and mobility preferences (which mode of transportation to use).

\addtolength{\textheight}{-12cm}   

\bibliographystyle{IEEEtran}
\bibliography{references}

\end{document}